\documentclass[11pt,a4paper]{article}
\pdfoutput=1 
\usepackage[utf8]{inputenc}
\usepackage[margin=1in]{geometry}

\usepackage{amsmath}
\usepackage{amsfonts}
\usepackage{amssymb}
\usepackage{amsthm}
\usepackage{mathtools}
\usepackage{stmaryrd}
\usepackage{commath}
\usepackage{xfrac}

\usepackage{lmodern}
\usepackage{fixcmex}
\usepackage[T1]{fontenc}
\usepackage{bm}

\usepackage[english]{babel}
\usepackage{csquotes}
\usepackage{url}
\usepackage[pdfencoding=auto,psdextra,hidelinks]{hyperref}
\usepackage{color}

\usepackage{float}
\usepackage{wrapfig}
\usepackage{booktabs}
\usepackage{subcaption}
\usepackage{authblk}

\usepackage{tikz}
\usetikzlibrary{positioning, calc}
\usepackage{pgfplots}
\pgfplotsset{width=10cm,compat=1.9}
\usepgfplotslibrary{fillbetween}
\usetikzlibrary{arrows.meta}
\usetikzlibrary{backgrounds}

\newcommand{\Norm}[1]{\ensuremath{\left\| #1 \right\|}}

\newcommand{\eps}{\ensuremath{\varepsilon}}

\newcommand{\RR}{\ensuremath{\mathbb{R}}}

\newcommand{\zob}{\ensuremath{\{0,1,\bot\}}}
\newcommand{\ncube}{\ensuremath{[0,1]^n}}

\newcommand{\NP}{\ensuremath{\mathrm{NP}}}

\newcommand{\ExistR}{\ensuremath{\exists\RR}}
\newcommand{\PPAD}{\ensuremath{\mathrm{PPAD}}}
\newcommand{\TFNP}{\ensuremath{\mathrm{TFNP}}}
\newcommand{\FIXP}{\ensuremath{\mathrm{FIXP}}}
\newcommand{\LinearFIXP}{\ensuremath{\mathrm{Linear\textrm{-}FIXP}}}

\newcommand{\PURIFY}{\ensuremath{\mathsf{PURIFY}}}

\newcommand{\NAND}{\ensuremath{\mathsf{NAND}}}

\newcommand{\ETR}{\textup{\textsc{ETR}}}

\newcommand{\QUAD}{\textup{\textsc{Quad}}}
\newcommand{\FEAS}{\textup{\textsc{4Feas}}}

\newcommand{\PureCircuit}{\textup{\textsc{Pure-Circuit}}}

\newcommand{\epsCDSClearing}{\ensuremath{\eps}\textup{\textsc{-CDS-Clearing}}}
\newcommand{\CDSHasClearing}{{\textup{\textsc{CDS-HasClearing}}}}
\newcommand{\CDSCanDefault}{{\textup{\textsc{CDS-CanDefault}}}}
\newcommand{\CDSCanSurvive}{{\textup{\textsc{CDS-CanSurvive}}}}

\newcommand{\trunczo}[1]{\ensuremath{\left\llbracket #1 \right\rrbracket}}
\newcommand{\dec}{\ensuremath{\operatorname{dec}}}

\newtheorem{theorem}{Theorem}
\newtheorem{proposition}{Proposition}
\newtheorem{lemma}{Lemma}
\newtheorem{corollary}{Corollary}

\newtheorem{definition}{Definition}

\begin{document}

\title{Improved Hardness Results for the Clearing Problem in\\Financial Networks with Credit Default Swaps}

\author{Simon Dohn}
\author{Kristoffer Arnsfelt Hansen}
\author{Asger Klinkby}
\affil{Aarhus University, Denmark}
\date{}
\maketitle

\begin{abstract}
  We study computational problems in financial networks of banks
  connected by debt contracts and credit default swaps (CDSs). A main
  problem is to determine \emph{clearing} payments, for instance right
  after some banks have been exposed to a financial shock. Previous
  works have shown the $\eps$-approximate version of the problem to be
  $\PPAD$-complete and the exact problem $\FIXP$-complete. We show
  that $\PPAD$-hardness hold when $\eps\approx 0.101$, improving the
  previously best bound significantly. Due to the fact that the
  clearing problem typically does not have a unique solution, or that
  it may not have a solution at all in the presence of default costs,
  several natural decision problems are also of great interest. We
  show two such problems to be $\ExistR$-complete, complementing
  previous $\NP$-hardness results for the approximate setting.
\end{abstract}

\section{Introduction}
The worlds financial systems form large networks in which financial
entities such as banks are closely interconnected by contracts. This
results in nontrivial dependencies between banks, where a default of a
single bank may affect large parts of the financial system. The
prevalent use of financial \emph{derivative} contracts, where a debt
obligation depends on other contracts introduces additional
complicated dependencies in the system.

A way to assess the complexity arising in such \emph{financial
  networks} is to study computational problems that arise in these
networks from the viewpoint of algorithms and computational
complexity. In this paper we consider the model of financial networks
defined by Schuldenzucker, Seuken
and~Battiston~\cite{SchuldenzuckerSSB2020-Default-Ambiguity}. This
models a static situation involving a number of financial entities
(e.g.\ banks) that have already entered (simple) debt contracts and
credit default swaps (CDSs) among each other. This is an extension of
the basic model of Eisenberg and
Noe~\cite{EisenbergNoe2001-Systemic-Risk} with CDSs as well as the
addition of default costs introduced to the model by Rogers and
Veraart~\cite{RogersVeraart2013-Failure-Rescue}.

After an event where a number of banks have experienced a financial
shock on their assets, some of these may become unable to meet their
obligations, leading them into default. This may in turn trigger
additional events in the financial network eventually causing a ripple
effect. This leads to the computational problem of evaluating all
contracts of the network simultaneously to find clearing payments
between the banks. More precisely we consider the problem of computing
a \emph{clearing recovery rate vector}, that describes for each bank
the fraction of its liabilities it is able to pay, and as a
consequence also which banks are in default.  This problem has been
considered in several prior
works. Schuldenzucker~et~al.~\cite{SchuldenzuckerSSB2020-Default-Ambiguity}
showed that, under mild assumptions, a clearing recovery rate vector
always exists for the setting where no costs are incurred due to the
process of bankruptcy, thereby turning the computational problem into
a total search problem. In a different work,
Schuldenzucker~et~al.~\cite{SchuldenzuckerSB2019-Complexity} proved
that an \eps-approximate version of the clearing problem is
\PPAD-complete and Ioannidis, Keijzer and
Ventre~\cite{IoannidisKV2022-Strong-Approximations} subsequently
proved that computing an (exact) clearing recovery rate vector is
\FIXP-complete. From an algorithmic viewpoint essentially nothing is
known. Indeed, one may observe that a trivial solution gives a
$\frac{1}{2}$-approximate solution and it is an open problem to find a
polynomial time approximation algorithm with a better guarantee. The
result of Schuldenzucker~et~al.\ proved \PPAD-hardness for a small
(unspecified) constant~$\eps>0$. From a practical perspective,
\PPAD-hardness for a miniscule constant $\eps>0$ may be less
concerning, and it is therefore of interest to determine for how large
values of $\eps$ the problem remains \PPAD-hard. Recently,
Ioannidis~et~al.~\cite{IoannidisKV2023-PPAD} considerably strengthened
the result of Schuldenzucker~et~al.\ and proved \PPAD-hardness for
$\eps\leq \frac{3-\sqrt{5}}{16} \approx 0.048$. We improve this
further and show \PPAD-hardness for
$\eps \leq 5 - 2 \sqrt{6} \approx 0.101$. Like Ioannidis~et~al.\ we
obtain our result by reduction from the problem \PureCircuit, recently
introduced Deligkas, Fearnley, Hollender and
Melissourgos~\cite{DeligkasFHM22-Pure-Circuit}, which is a
``generalized circuit'' problem operating on the domain
$\zob$. Different from existing reductions (also to problems unrelated
to financial networks) from \PureCircuit\ we decode values of the
interval~$[0,1]$ to elements of $\zob$ in an \emph{asymmetric} way,
more suited to exploit the characteristics of CDSs in financial
networks. This in turn allows for a simpler reduction that
additionally yields a stronger result.

A clearing recovery rate vector of a financial network will typically
not be unique leading to the problem of \emph{selecting} one of
these. Simple questions a regulator might be interested in knowing the
answer to, could be whether a particular bank or set of banks are
bound to default, whether a default may be prevented, or whether they
are guaranteed to avoid default. When default costs are present, a
clearing recovery rate vector may not exist at all. This means that
computing a clearing recovery rate vector (in case it exists) may be
even harder in the presence of default costs.

These considerations lead to several natural decision problems and
Schuldenzuker~et~al.~\cite{SchuldenzuckerSB2019-Complexity} proved
that these are \NP-hard, in fact even for \emph{gap} versions of the
problems relevant for the setting of \eps-approximate clearing. We
complement these results by showing $\ExistR$-completeness for the
problem of deciding if a given bank may avoid default and, in the
setting of default costs (but on external assets only), for the
problem of whether a clearing recovery rate exists. Having default
costs only on external assets could model a setting where external
assets are illiquid in comparison to obligations within the financial
system, and having to liquidate these in response to a financial shock
may incur a considerable loss.

These $\ExistR$-hardness results have direct consequences from an
algorithmic perspective. Namely, solving these problems in the setting
of exact clearing solutions must involve algorithmic techniques
capable of solving general systems of polynomial equations. For the
complexity class $\ExistR$ we have the relation
$\NP \subseteq \ExistR$, and this inclusion is generally conjectured
to be strict. This suggests that the decision problems for exact
clearing are considerably harder than their approximate counterparts.

\section{Preliminaries}

\subsection{Complexity Classes}

Our results are concerned with the complexity classes $\PPAD$ and
$\ExistR$, but it will also be relevant to introduce the complexity
class $\FIXP$. The classes $\PPAD$ and $\FIXP$ are classes of total
search problems, whereas the class $\ExistR$ is a class of decision
problems.  We introduce the classes briefly below and refer to the
references for further details.

The class $\PPAD$ was defined by
Papadimitriou~\cite{Papadimitriou1994-TFNP} as the subclass of $\TFNP$
consisting of all problems reducible to the problem
\textsc{End-Of-Line}.  We shall however not use this definition
directly, and hence we do not define this problem formally. Instead,
to prove $\PPAD$-hardness we make use of the $\PureCircuit$ problem,
recently introduced by Deligkas, Fearnley, Hollender and
Melissourgos~\cite{DeligkasFHM22-Pure-Circuit} and for
$\PPAD$-membership we make use of general results about computing
$\eps$-almost fixed points of polynomially continuous and polynomially
computable functions. We introduce the $\PureCircuit$ problem in
detail in Section~\ref{sec:PPAD-hardness-approx}.

The class $\FIXP$ is a class of total real valued search problems
defined by Etessami and
Yannakakis~\cite{EtessamiYannakakis2010-FIXP}. It consists of all
problems reducible, by a so-called SL-reduction, to the problem of
finding a fixed point of a continuous function $F \colon D \to D$,
where $D \subseteq \RR^n$ convex set. The class $\FIXP$ allows several
different equivalent definitions in how $F$ and $D$ should be
described. A simple such definition is to let $D=\ncube$ and the
function $F$ be given by an algebraic circuit with $n$~inputs and
$n$~outputs over the basis $\{+,*,\max\}$ and allowing use of
arbitrary rational constants.

If we restrict the basis to be $\{+,*c,\max\}$, where $*c$ refers to
multiplication by a constant, we obtain the class $\LinearFIXP$. The
class can be viewed as a class of discrete total search problems and
with this interpretation Etessami and Yannakakis proved that
$\PPAD=\LinearFIXP$.

The class $\ExistR$ was defined formally by Schaefer and
Štefankovič~\cite{Schaefer2009-Complexity-Some-Geometric,SchaeferStefankovic2017-Nash-ETR}
and Bürgisser and
Cucker~\cite{BurgisserCucker2009-Exotic-Quantifiers}. One may think of
the class $\ExistR$ as having a relationship to $\NP$ similarly to the
relationship of $\FIXP$ to $\PPAD$. Schaefer and Štefankovič defined
the class $\ExistR$ to have the decision problem for the existential
theory of the reals, $\ETR$, as its complete problem, whereas
Bürgisser and Cucker defined the class as the constant free Boolean
part of the analogue class $\NP_\RR$ to $\NP$ in the Blum-Shub-Smale
model of computation. For proving $\ExistR$-membership the latter
definition is most convenient.

The standard \ExistR-complete problem is the problem \QUAD\ of
determining whether a system of quadratic polynomials in~$n$ variables
has a solution. Schaefer~\cite[Lemma~3.9]{Schaefer2013-Realizability}
proved $\QUAD$ remains $\ExistR$-hard even with a promise that the
system has a solution in the unit~$n$-ball whenever the system has any
solution. We need the analogous statement for the unit $n$-cube, which
we denote by $\QUAD(\ncube)$.
\begin{proposition}
  $\QUAD(\ncube)$ is \ExistR-hard.
  \label{prop:QUADcube}
\end{proposition}
This result can either easily be derived from the result of Schaefer
or follows as a direct corollary of
\cite[Proposition~2]{Hansen2019-Real-Complexity}. Next we consider the
problem \FEAS. Here we are given a single degree~4 polynomial $p$
in~$n$ variables and the task is to decide if $p$ has a root. We let
$\FEAS(\ncube)$ refer to the promise version of the problem where a
root in $\ncube$ is guaranteed to exist in case $p$ has any root.
\begin{corollary}
  $\FEAS(\ncube)$ is \ExistR-hard.
  \label{cor:FEAScube}
\end{corollary}
\begin{proof}
  We simply observe that the standard reduction \QUAD\ to \FEAS\
  preserves the promise. Namely, given quadratic polynomials
  $p_1,\dots,p_k$ in~$n$ variables we define the degree~4 polynomial
  $p(x) = \sum_{i=1}^k (p_i(x))^2$. It then follows that $x$ is a
  solution to the system of equations $p_1(x)=\dots=p_k(x)=0$ if and
  only if $p(x)=0$.
\end{proof}

\subsection{Financial Networks}
We consider the model of financial networks defined by Schuldenzucker,
Seuken
and~Battiston~\cite{SchuldenzuckerSSB2020-Default-Ambiguity}. This
models a static situation of banks that have already entered debt
contracts and CDSs among each other that must now all be evaluated
simultaneously.

\paragraph{Banks and Contracts}
A \emph{financial network} is a given by a finite set $N$ of
$n=\abs{N}$ \emph{banks}. Each bank~$i \in N$ holds an amount of
\emph{external assets} $e_i \geq 0$. Debt contracts are given by
$c_{i,j} \geq 0$ for distinct banks $i,j \in N$ and CDSs are given by
$c_{i,j}^k \geq 0$ for distinct banks $i,j,k \in N$. The numbers
$c_{i,j}$ and $c_{i,j}^k$ are commonly referred to as the
\emph{notionals} of the contracts. For notational convenience we let
$c_{i,i}=0$ and $c_{i,i}^j=c_{i,j}^i=c_{i,j}^j=0$ for all $i,j \in N$.

The contracts specify obligations for the \emph{writer} (debtor) to
pay an amount of money, the \emph{liability}, to the \emph{holder} of
the contract. A bank unable to fulfill all its obligations is \emph{in
  default}. The \emph{recovery rate} $r_i \in [0,1]$ of bank $i \in N$
is the fraction of liabilities it is able to pay. A bank is thus in
default if and only if $r_i<1$. A bank which is not in default is also
said to be \emph{solvent}.

When $c_{i,j}>0$, bank~$i$ is obligated to pay the (unconditional)
liability $c_{i,j}$ to bank~$j$. When $c_{i,j}^k>0$, bank~$i$ is
obligated to pay the (conditional) liability $(1-r_k)c_{i,j}^k$ to
bank~$k$, based on the recovery rate of the \emph{reference}
bank~$k$. As specified above, a bank can only take one of the two
roles of a debt contract or one of the three roles of a CDS.

A CDS allows an insurance on debt contracts: If bank~$j$ holds both a
simple debt contract written by bank~$k$ as well as a CDS with the
same notional written by bank~$i$ with reference~$k$, then bank $j$ is
guaranteed to receive the amount of the debt contract in the event
that bank~$k$ defaults (as long as bank~$i$ does not default as well).

\paragraph{Liabilities, Payments, and Assets}

Let $r \in \ncube$ be a given vector of recovery rates. The
\emph{liability} of bank~$i$ to bank~$j$ is given by
$
  l_{i,j}(r) = c_{i,j} + \sum_{k\in N}(1-r_k)c_{i,j}^k
$
and the
\emph{total liability} of bank~$i$ is
$l_i(r) = \sum_{j\in N} l_{i,j}(r)$.

The total liability is converted into a total payment according to the
principles of \emph{absolute priority} and \emph{limited liability},
meaning that banks must pay their liabilities in full if possible and
never pay more than the entirety of their assets (subtracted possible
default costs). Liabilities to individual banks are converted into
payments according to the principle of \emph{proportionality}, meaning
that the total payment is divided according to each individual
liability's share of the total liability. The \emph{payment} from
bank~$i$ to bank~$j$ is thus given by
$p_{i,j}(r) = r_i\cdot l_{i,j}(r)$.

The assets of bank~$i \in N$ before \emph{default costs} is given by
\begin{equation}
  a_i(r) = e_i + \sum_{j\in N} p_{j, i}(r) \enspace .
\end{equation}
Default costs are modeled by two parameters
$\alpha, \beta \in [0,1]$. When in default, a bank~$i$ is only able to
recover an $\alpha$-fraction of its external assets and a
$\beta$-fraction of incoming payments. The assets of bank~$i$ after
default costs is then given by
\begin{equation}
  a'_i (r) = \alpha e_i + \beta \sum_{j\in N} p_{j, i}(r) \enspace .
\end{equation}

\paragraph{Clearing Recovery Rates} Since the liabilities and assets
of the banks are given in terms of presumed recovery rates and in turn
also define recovery rates we are left with a fixed point problem.
\begin{definition}[Clearing recovery rate vector]
  A recovery rate vector $r \in \ncube$ is \emph{clearing} if it is a
  fixed point of the function $F \colon \ncube \to \ncube$ given by
  \begin{equation}
    F(r)_i =
    \begin{cases}
      1 & \text{ if } a_i(r) \geq l_i(r) \\
      \frac{a'_i(r)}{l_i(r)} & \text{ if } a_i(r) < l_i(r)
    \end{cases}
    \label{eq:F-function}
  \end{equation}
\label{def:Clearing}
\end{definition}

The function $F$ is introduced by Schuldenzucker, Seuken, and
Battiston as the \emph{update function}. This name is well-justified
for financial networks where every CDS is \emph{covered}, in which
case iterating the function $F$, starting from the all ones vector,
converges to a clearing recovery rate vector that maximizes the
recovery rate of all banks
simultaneously~\cite[Corollary~1]{SchuldenzuckerSSB2020-Default-Ambiguity}. In
general, iterating the function $F$ need not converge even when
clearing recovery rate vectors exist.

When the financial network has default costs, i.e.\ when either
$\alpha < 1$ or $\beta < 1$, simple and natural examples show that a
clearing recovery rate vector is not guaranteed to
exist~\cite[Theorem~1]{SchuldenzuckerSSB2020-Default-Ambiguity}. But
even in the case when $\alpha=\beta=1$ there exist (somewhat
artificial) examples (see Appendix~\ref{sec:degeneracy}) of financial
networks without a clearing recovery rate arising due to the
discontinuity in the definition of the function $F$. To avoid this we
shall as in previous works impose a \emph{non-degeneracy} assumption.
\begin{definition}[Non-degeneracy]
  A financial network is called \emph{non-degenerate} if for all
  $i \in N$, either $e_i>0$ or there exists $j \in N$ such that
  $c_{i,j}>0$.
  \label{def:non-degeneracy}
\end{definition}
In constructions it is useful to have banks that have written no
contracts but hold one or more contracts. In this case it is without
loss of generality to supply the bank with external assets, since the
bank is solvent by definition.

Note that for a non-degenerate financial network we are guaranteed
that $\max(a_i(r),l_i(r))$ is bounded from below by a fixed positive
constant. As noted by Ioannidis, Keijzer and
Ventre~\cite{IoannidisKV2022-Strong-Approximations}, in case
$\alpha=\beta=1$, a clearing recovery rate vector may thus be found as
a fixed point of the continuous function
$f \colon \ncube \to \ncube$ given by
\begin{equation}  
  f(r)_i = \frac{a_i(r)}{\max(a_i(r),l_i(r))} \enspace ,
    \label{eq:f-function}
\end{equation}
guaranteed to exist by Brouwer's fixed point theorem. Note that when
$l_i(r)\neq 0$ we have that $F(r)_i=\min\left(1,\frac{a_i(r)}{l_i(r)}\right)=f(r)_i$.

\paragraph{Notation}
For $x \in \RR$, we denote by $\trunczo{x}$ the \emph{truncation} of
$x$ to the interval $[0,1]$, i.e.\ $\trunczo{x} =
\min(1,\max(0,x))$. For $\eps>0$ we write $x=y\pm\eps$ to mean that
$x \in [y-\eps,y+\eps]$. Both notations extends naturally to vectors,
i.e.\ $(\trunczo{x})_i=\trunczo{x_i}$ for all $i$, and $x=y\pm\eps$ if
$x_i=y_i\pm\eps$ for all $i$.

\paragraph{Diagrams of Financial Networks} We may conveniently
represent a financial network as a labeled directed graph. The nodes
are given by the set of banks $N$. The node $i \in N$ is labeled by
the external assets $e_i$. When $c_{i,j}>0$, the graph has an arc from
$i$ to $j$ labeled by the notional $c_{i,j}$. When $c_{i,j}^k>0$, the
graph has an arc from $i$ to $j$ labeled by the notional $c_{i,j}^k$
as well as the reference bank $k$. When illustrating these graphs as
diagrams we label external assets $e_i$ in a box drawn on top of the
node of the bank~$i$. This label may be omitted when $e_i=0$. A debt
contract is drawn as a blue arc and a CDS as an orange arc, both
labeled by the notional. Instead of labeling the arc corresponding to
a CDS also by the reference bank, we connect the reference bank and the
arc by a dashed line. An example is given in
Figure~\ref{fig:example-network}. 

\begin{figure}[h!]
  \centering
  \begin{tikzpicture}
  
  \begin{scope}[every node/.style={circle, thick, draw}]
    \node (A) at (0,2) {$A$};
    \node (D) at (2,0) {$D$};
    \node (B) at (2,2) {$B$};
    \node (C) at (0,0) {$C$};
  \end{scope}
  
  \begin{scope}[ every edge/.style={draw=blue,very thick}]
    \path [->] (A) edge node[below] {$2$} (B);
    \path [->] (D) edge node[right] {$3$} (B);
  \end{scope}
  
  \begin{scope}[every edge/.style={draw=orange,very thick}]
    \path [->] (C) edge node[below] {$4$} (D);
  \end{scope}
  
  \begin{scope}[every edge/.style={draw=orange,dashed,very thick}]
    \path[-] (A) edge ($(D) !.5! (C)$);
  \end{scope}

  \begin{scope}[xshift=0.4cm, yshift=-0.4cm, every node/.style={rectangle, draw=black, fill=white}]   
    \node (eA) at (0,2) {$1$};
    \node (eB) at (2,2) {$1$};
    \node (eD) at (2,0) {$0$};
    \node (eC) at (0,0) {$4$};
  \end{scope}
\end{tikzpicture}    

  \caption{Example diagram of financial network.}
  \label{fig:example-network}
\end{figure}
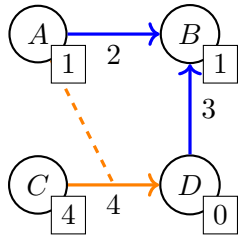

When there are no default cost, i.e.\ when $\alpha=\beta=1$, the
unique clearing recovering vector in this example is
$(r_A,r_B,r_C,r_D)=(1/2,1,1,2/3)$. Suppose now that $\alpha=1/2$ and
$\beta=2/3$. In this case the unique clearing recovering vector is instead
$(r_A,r_B,r_C,r_D)=(1/4,1,1 ,1)$. Thus introducing default costs made
bank $D$ solvent.  Note also that the value of $\beta$ is not relevant
in the example.

\paragraph{Source and Sink Banks}
When constructing financial networks it is convenient to have banks
that are sources and sinks of the associated directed graph and that
these banks are \emph{far} from default in any
circumstance~\cite{SchuldenzuckerSB2019-Complexity}. The role of these
banks are therefore exclusively to act as writers and holders,
respectively, of contracts, and we refer to these as \emph{source
  banks} and \emph{sink banks}. More precisely, we shall assume that
the assets of source and sink banks are always at least double the
amount of their liabilities. This trivially holds for sink banks and
for a source bank $s$ we can simply assume that
$e_s \geq 2 \left(\sum_{j \in N} c_{s,j}+\sum_{k \in
    N}c_{s,j}^k\right)$.  To ensure the banks satisfy the
non-degeneracy assumption of Definition~\ref{def:non-degeneracy}, we
shall in addition assume that a source bank $s$ holds a debt contract
with a sink bank $t$ of notional of~1 and we assume that a sink bank
$t$ has external assets $e_t=1$.

In diagrams of financial networks we reserve the labels $s$ and $t$
for source and sink banks, respectively, possibly with subscripts. For
the purpose of clarity in the diagrams we do not label source and sink
banks with their external assets and we do not draw the contracts
directly from source banks to sink banks used to ensure that the
non-degeneracy condition holds.

\subsection{Computation by Financial Networks}
The complexity of computational problems about financial networks
arises from the ability of the networks to perform
computation. Consider the financial network given in
Figure~\ref{fig:basic-gadget} where bank~$v$ holds a CDS with
reference~$u$ of notional~$c_{s,v}^u$ and is also the writer of a contract
of notional~1. We assume that the recovery rate of bank~$u$ is
determined by assets and liabilities not illustrated. We now have the
relationship
\begin{equation}
  r_v = \min\left(1, c_{s,v}^u (1 - r_u)\right) \enspace ,
\end{equation}
which one can see resembles the voltage transfer characteristics of an
inverter, indicating its usefulness as a basic building block.

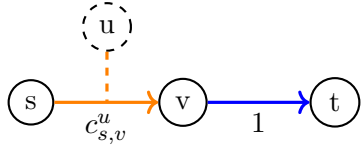
\begin{figure}[h!]
    \centering
    \begin{tikzpicture}
  
  \begin{scope}[every node/.style={circle,thick,draw}]
    \node (u) at (1,2) [dashed] {u};  no dashes yet
    \node (s) at (0,1) {s};
    \node (v) at (2,1) {v};
    \node (t) at (4,1) {t};
  \end{scope}
  
  \begin{scope}[ every edge/.style={draw=blue,very thick}]
    \path [->] (v) edge node[below] {$1$} (t);
  \end{scope}
  
  \begin{scope}[every edge/.style={draw=orange,very thick}]
    \path [->] (s) edge node[below] {$c_{s,v}^u$} (v);
  \end{scope}
  
  \begin{scope}[every edge/.style={draw=orange,dashed,very thick}]
    \path[-] (u) edge ($(s) !.5! (v)$);
  \end{scope}
\end{tikzpicture}    

    \caption{Financial network gadget with input bank $u$ and output bank $v$.}
    \label{fig:basic-gadget}
\end{figure}

Another source of computational power of the model of financial
networks comes from the fact that payments are obtained from
liabilities by \emph{multiplication} of the recovery rate. When
considering exact clearing recovery rate vectors this enables
financial networks to evaluate polynomials
(see Section~\ref{sec:compute-polys}) and general algebraic
circuits~\cite{IoannidisKV2022-Strong-Approximations}.

\paragraph{Financial Network Gadgets}
Generalizing on the example of Figure~\ref{fig:basic-gadget} we
consider small financial networks, which we refer to as financial
network \emph{gadgets}, with designated \emph{input} and \emph{output}
banks. An input bank will be the output bank of another gadget (and
hence illustrated with a dashed circle). Non-input banks will
typically have no external assets and have their recovery rates depend
functionally on the recovery rates of the input banks.

\section{\PPAD\ hardness of Approximate Clearing}
\label{sec:PPAD-hardness-approx}

It is not a priori obvious what would constitute a reasonable notion
of approximation of clearing recovery rate vectors. Ideally such an
approximation should enable a reasonable resolution of contracts. In
particular, when computing payments from approximations to recovery
rates, is is crucial that money is not missing in the system to be
able to do so. On the other hand, from the perspective of proving
computational hardness results this becomes a lesser concern, and such
results are meaningful as long as the conditions imposed by the notion
of approximation can be viewed as necessary conditions.

For approximation we focus on the case of financial networks without
default costs. A candidate notion of approximation is that of an
\eps-almost fixed point of the function $f$, i.e.\ $r \in \ncube$ such
that $\Norm{f(r)-f}_\infty \leq \eps$, and this is equivalent to the
notion of approximation considered by
Schuldenzucker~et~al.~\cite{SchuldenzuckerSB2019-Complexity}. For
non-degenerate financial networks we have that the problem of
computing \eps-almost fixed point of the function $f$ is contained in
\PPAD\ by a general result of Etessami and
Yannakakis~\cite[Proposition~2.2]{EtessamiYannakakis2010-FIXP}.

A weakness\footnote{In addition, using this definition would have
  implications for proving hardness of approximation, since it would
  not be possible to guarantee that a source bank is solvent no matter
  how many external assets it has. Our result, as well as the result
  of Ioannidis~et~al.~\cite{IoannidisKV2023-PPAD}, could be adapted to
  this weaker definition, but would result in \PPAD-hardness for
  smaller values of \eps.} of this definition is that it allows for a
bank that is very far from being in default to have a recovery rate
of~$1-\eps$, which in turn would trigger CDSs that could have
significant impact unless $\eps$ is negligible.
Ioannidis~et~al.~\cite{IoannidisKV2023-PPAD} imposed the additional
condition that a bank~$i$ for which
$e_i > \sum_{j \in N} c_{s,j}+\sum_{k \in N}c_{s,j}^k$ must have
recovery rate~$r_i=1$. This condition may be handled during a
preprocessing step, which then gives \PPAD-membership also for this
modified definition. While it is clearly justified to declare such a
bank solvent, it seems a somewhat artificial to restrict the solvency
condition to banks that are \emph{guaranteed} to be solvent due to
sufficient external assets. It arguably appears more natural to extend
the condition to say that a bank with more assets than liabilities in
a given situation must have recovery rate~$1$, i.e.\ to enforce the
condition $a_i(r) > l_i(r) \implies r_i = 1$. It is however not clear
whether $\PPAD$-membership holds with this condition, due to the sharp
transition between solvency and default. Since we aim for a reasonable
notion of approximation we shall relax the condition to only require a
recovery rate of~$1$ when a bank clearly has more assets than
liabilities. We therefore propose the following notion of approximate
clearing recovery rate vectors.
\begin{definition}[\eps-approximate clearing recovery rate vector]
  For a financial network without default costs, a recovery rate
  vector $r \in \ncube$ is an \eps-approximate clearing recovery rate
  vector if $\Norm{f(r)-r}_\infty \leq \eps$ and
  $a_i(r) \geq (1+\eps)l_i(r) \implies r_i=1$ for all $i\in N$.
\label{def:eps-approx}
\end{definition}
We shall refer to the problem of computing an \eps-approximate
clearing recovery rate vector in non-degenerate financial networks
without default costs as the \epsCDSClearing\ problem. We give the
simple proof of \PPAD-membership for \epsCDSClearing\ with respect to
Definition~\ref{def:eps-approx} in
Appendix~\ref{sec:PPAD-membership}. Let us also note that the
\PPAD-hardness result of Ioannidis~et~al.~\cite{IoannidisKV2023-PPAD}
still holds under the weaker assumption above. Our result is the
following.
\begin{theorem}
  The problem \epsCDSClearing\ is $\PPAD$-hard for
  $\eps \leq 5 - 2 \sqrt{6} \approx 0.101$.
  \label{thm:epsCDSCLearing-PPAD}  
\end{theorem}
While still being very far from the (trivial) upper bound of
$\eps=\frac{1}{2}$, the result is a significant quantitative
improvement upon the previous work of
Ioannidis~et~al.~\cite{IoannidisKV2023-PPAD} that obtained the same
result for $\eps\leq \frac{3-\sqrt{5}}{16} \approx 0.048$. Both of
these bounds on $\eps$ are clearly in the range of practical
relevance.

\subsection{Pure-Circuit}
To show $\PPAD$-hardness of the approximate clearing problem, we will
be reducing from the problem $\PureCircuit$, like the previous work of
Ioannidis~et~al.~\cite{IoannidisKV2023-PPAD}. The $\PureCircuit$
problem was recently introduced and proved be be complete for $\PPAD$
by Deligkas~et~al.~\cite{DeligkasFHM22-Pure-Circuit}, and with this
the same authors could greatly simplify and improve previous
$\PPAD$-hardness results for a range of problems.

The $\PureCircuit$ problem is a type of ``generalized circuit''
problem~\cite{ChenDT2009-Settling-two-player-Nash}, which means that
the instance is given as a directed graph, the \emph{circuit}, with
nodes being \emph{gates} computing a specific function from its inputs
to its outputs, but where the circuit itself does not have designated
inputs or outputs. The circuit instead has cycles, which in turn means
that the problem is essentially a constraint satisfaction problem. The
problem is however defined such that a solution is guaranteed to
exist.

The gates of a $\PureCircuit$ instance operate on the domain
$\zob$. We may view this as an abstraction of a value in the interval
$[0,1]$ where $\bot$, the ``garbage'' value, refer to any value in the
open interval $(0,1)$. The usual Boolean gates are readily extended to
the domain~\zob. We give the definition for the (complete) function
$\NAND$.
\begin{definition}
  The \NAND\ gate with inputs $u,v \in \zob$ and output $w \in \zob$
  is satisfied if and only if $(u=1 \wedge v=1 \implies w=0)$ and
  $(u=0 \vee v=0 \implies w=1)$.
  \label{def:NAND}
\end{definition}
The novel \PURIFY\ gate has two outputs and guarantees a ``pure bit''
on one of these.
\begin{definition}
  The \PURIFY\ gate with input $u \in \zob$ and outputs $v,w\in\zob$
  is satisfied if and only if $(u \neq \bot \implies v=u \wedge w=u)$ and
  $(v \neq \bot \vee w \neq \bot)$.
  \label{def:PURIFY}
\end{definition}
The \PureCircuit\ problem is \PPAD-complete for circuits with \PURIFY\
gates together with a set of functionally complete Boolean gates. In
particular we have the following.
\begin{theorem}[Deligkas~et~al.~\cite{DeligkasFHM22-Pure-Circuit}]
  \PureCircuit\ is \PPAD-complete for circuits with \NAND\ and
  \PURIFY\ gates.
\end{theorem}

\subsection{Proof of Theorem~\ref{thm:epsCDSCLearing-PPAD}}
We make a reduction from $\PureCircuit$ with \NAND\ and \PURIFY\ gates
to \epsCDSClearing\ by constructing financial network gadgets for both types of
gates, where the inputs and outputs of the gates correspond to input
and output banks. Keeping the network topology of these gadgets fixed,
each gadget will involve a few parameters in the form of notionals of
CDSs. It turns out that the \NAND\ gate is the most restrictive
towards maximizing $\eps$, and we thus give a full analysis of our
choice of parameters, showing it is optimal.

Besides the gadget we require a way to to map a clearing recovery rate
vector to a solution of the \PureCircuit\ instance by means of a
\emph{decoding function} $\dec \colon [0,1] \to \zob$. The function
$\dec$ depends on two additional parameters $\gamma,\delta \geq 0$
satisfying $\gamma+\delta<1$ and is defined as follows.
\begin{equation}
  \dec(r) = \begin{cases}
    0    & \text{ if } r \in [0,\gamma]\\
    \bot & \text{ if } r \in (\gamma,1-\delta)\\
    1    & \text{ if } r \in [1-\delta,1]
  \end{cases} \enspace .
\end{equation}
The characteristics of a CDS, as illustrated by
Figure~\ref{fig:basic-gadget} has the consequence that it is
advantageous to have the intervals in $[0,1]$ (of length $\gamma$ and
$\delta$, respectively) that are mapped to the endpoints by $\dec$ be
\emph{asymmetric}. In particular, for our setting of parameters, we will
have $\gamma = 3-\sqrt{6} \approx 0.551$ and
$\delta=5-2\sqrt{6} \approx 0.101$, but for now we leave them unfixed.

\paragraph{\NAND\ gadget}
Our financial network gadget for the \NAND\ gate is given in
Figure~\ref{fig:NAND}, with parameters $c_1$ and $c_2$ still to be
determined. This gadget was already defined by
Schuldenzucker~et~al.~\cite{SchuldenzuckerSB2019-Complexity}, but in a
different setting and with different parameters and analysis.

\begin{figure}[h]
    \centering
    \begin{tikzpicture}
  
  \begin{scope}[every node/.style={circle,thick,draw}]
    \node (u) at (0.7,3) [dashed] {$u$};
    \node (v) at (1.7,3) [dashed] {$v$};
    
    \node (s) at (0,1) {$s$};
    \node (w) at (2.5,1) {$w$};
    \node (t) at (4,1) {$t$};
  \end{scope}
  
  \begin{scope}[ every edge/.style={draw=blue,very thick}]
    \path [->] (w) edge node[below] {$1$} (t);
  \end{scope}
  
  \begin{scope}[every edge/.style={draw=orange,very thick}]
    \path [->] (s) edge [bend left] coordinate[pos=.25] (top) node[above] {$c_1$} (w);
    \path [->] (s) edge [bend right] coordinate[pos=.75] (bot) node[below] {$c_2$} (w);
  \end{scope}
  
  \begin{scope}[every edge/.style={draw=orange,dashed,very thick}]
    \path[-] (u) edge (top);
    \path[-] (v) edge (bot);
  \end{scope}
\end{tikzpicture}    

    \caption{\NAND\ gadget}
    \label{fig:NAND}
\end{figure}

For the gadget to operate correctly we require that for any
\eps-approximate clearing recovery rate vector, the values
$\dec(r_u)$, $\dec(r_v)$, and $\dec(r_w)$ must satisfy the \NAND\
gate. Having an $\eps$-approximate solution means that
$r_w=\trunczo{(1-r_u)c_1+(1-r_v)c_2}\pm \eps$. In particular,
\begin{equation}
  \begin{aligned}
r_w &\geq \min(1,(1-r_u)c_1) - \eps \\
r_w &\geq \min(1,(1-r_v)c_2) - \eps \\
r_w &\leq (1-r_u)c_1 + (1-r_v)c_2 + \eps
  \end{aligned}
\end{equation}
To satisfy Definition~\ref{def:NAND} we must ensure that that
$r_w \geq 1-\delta$ when either $r_u \leq \gamma$ or
$r_v \leq \gamma$, and that $r_w \leq \gamma$ when both
$r_u \geq 1-\delta$ and $r_v \geq 1-\delta$. This leads to the
requirements
\begin{align}
  1-\eps & \geq 1 - \delta \label{eq:NANDreq1}\\
  (1-\gamma)c_1 - \eps & \geq 1-\delta \label{eq:NANDreq2}\\
  (1-\gamma)c_2 - \eps & \geq 1-\delta \label{eq:NANDreq3}\\
  \delta c_1 + \delta c_2 + \eps & \leq \gamma \label{eq:NANDreq4}
\end{align}
To maximize $\eps$ we see that we should minimize $\delta$, and
(\ref{eq:NANDreq1}) dictates that $\delta=\eps$. Inequality
(\ref{eq:NANDreq4}) shows that we should minimize both $c_1$ and
$c_2$, and (\ref{eq:NANDreq2}) and (\ref{eq:NANDreq3}) dictates that
$c_1=c_2=\frac{1}{1-\gamma}$. Letting (\ref{eq:NANDreq4}) hold with
equality and solving for $\epsilon$ yields
$\eps=\frac{\gamma(1-\gamma)}{3-\gamma}$. The maximum for
$\gamma \in [0,1]$ is $\eps=5-2\sqrt{6} \approx 0.101$ obtained for
$\gamma=3-\sqrt{6}\approx 0.551$, and this is a valid choice for $\gamma$ since
$\gamma+\delta = 8 - \sqrt{6} \approx 0.652 < 1$.

\paragraph{\PURIFY\ gadget} 
Our financial network gadget for the \PURIFY\ gate is given in
Figure~\ref{fig:PURIFY}, with $\gamma$ and $\eps$ specified as
above. The parameter $\phi$ is still to be specified, while we let
$\eta=\frac{1-\phi}{1-\gamma}+\varepsilon$. The gadget was introduced
by Ioannidis~et~al.~\cite{IoannidisKV2023-PPAD} with different
parameters and analysis. 

\begin{figure}[h!]
    \centering
    \begin{tikzpicture}
  
  \begin{scope}[every node/.style={circle, thick, draw}]
    \node (u) at (4,3) [dashed] {$u$};
    \node (s1) at (0,1) {$s_1$};
    \node (A) at (2,1) {$A$};
    \node (t1) at (4,1) {$t_1$};
    \node (s2) at (1,-0.4) {$s_2$};
    \node (v) at  (3,-0.4) {$v$};
    \node (t2) at (5,-0.4) {$t_2$};

    \node (s3) at (6+0,1) {$s_3$};
    \node (B) at  (6+2,1) {$B$};
    \node (t3) at (6+4,1) {$t_3$};
    \node (s4) at (6+1,-0.4) {$s_4$};
    \node (w) at  (6+3,-0.4) {$w$};
    \node (t4) at (6+5,-0.4) {$t_4$};
  \end{scope}
  
  \begin{scope}[ every edge/.style={draw=blue,very thick}]
    \path [->] (A) edge node[below] {$1$} (t1);
    \path [->] (v) edge node[below] {$1$} (t2);
    \path [->] (B) edge node[below] {$1$} (t3);
    \path [->] (w) edge node[below] {$1$} (t4);
  \end{scope}
  
  \begin{scope}[every edge/.style={draw=orange,very thick}]
    \path [->] (s1) edge node[below] {$\frac{1}{1 - \phi}$} (A);
    \path [->] (s2) edge node[below] {$\frac{1}{1 - \gamma}$} (v);
    \path [->] (s3) edge node[below] {$\frac{1}{1 - \gamma}$} (B);
    \path [->] (s4) edge node[below=3pt] {$\frac{1}{1 - \eta}$} (w);  \end{scope}
  
  \begin{scope}[every edge/.style={draw=orange,dashed,very thick}]
    \path[-] (u) edge ($(s1) !.5! (A)$);
    \path[-] (A) edge ($(s2) !.5! (v)$);
    \path[-] (u) edge ($(s3) !.5! (B)$);
    \path[-] (B) edge ($(s4) !.5! (w)$);
  \end{scope}
\end{tikzpicture}    

    \caption{\PURIFY\ gadget}
    \label{fig:PURIFY}
\end{figure}

For the gadget to operate correctly we require that for any
\eps-approximate clearing recovery rate vector, the values
$\dec(r_u)$, $\dec(r_v)$, and $\dec(r_w)$ must satisfy the \PURIFY\
gate. By Definition~\ref{def:PURIFY} this means that we must satisfy
the following conditions.
\begin{align}
  &\dec(r_u) = 0 \implies \dec(r_v)=0 \wedge \dec(r_w)=0 \label{eq:PURIFYreq0}\\
  &\dec(r_u) = 1 \implies \dec(r_v)=1 \wedge \dec(r_w)=1 \label{eq:PURIFYreq1}\\
  &\dec(r_v)\neq \bot \vee \dec(r_w)\neq \bot \label{eq:PURIFYreqbot}
\end{align}
It turns out that there is a small range of values for $\phi$ where
these holds true. Rather than making a complete analysis, we make the
arbitrarily choice of $\phi = \frac{7}{10}$ from this range, and
simply prove that this works. Having an $\eps$-approximate solution means that
\begin{align}
r_A&=\trunczo{\frac{1-r_u}{1-\phi}}\pm\eps \enspace ,&
r_v&=\trunczo{\frac{1-r_A}{1-\gamma}}\pm\eps \label{eq:PURIFY-A-v}\\
r_B&=\trunczo{\frac{1-r_u}{1-\gamma}}\pm\eps \enspace ,&
r_w&=\trunczo{\frac{1-r_B}{1-\eta}}\pm\eps \label{eq:PURIFY-B-w}
\end{align}

We start by showing that (\ref{eq:PURIFYreqbot}) holds. Let us first
assume $r_u \leq \phi$. From (\ref{eq:PURIFY-A-v}) we have that
$r_A \geq 1-\eps$ and then further that
\[
  r_v \leq \frac{\eps}{1-\gamma}+\eps = \gamma\frac{2-\gamma}{3-\gamma} \leq \gamma \enspace ,
\]
where the equality follows from our choice of $\gamma$ and $\eps$. This means that $\dec(r_v)=0$.  Suppose now that $r_u\geq \phi$. From
(\ref{eq:PURIFY-B-w}) we have that
$r_B \leq \frac{1-\phi}{1-\gamma}+\eps=\eta$ and then further that
$r_w \geq 1-\eps = 1-\delta$. This means that $\dec(r_w)=1$.

We proceed to show that (\ref{eq:PURIFYreq0}) and
(\ref{eq:PURIFYreq1}) hold. Note that since
$\gamma \leq \phi \leq 1 - \delta$, we have already shown that
$r_u \leq \gamma$ implies $r_v \leq \gamma$ and that
$r_u \geq 1- \delta$ implies $r_w \geq 1 - \delta$, which means that
$\dec(r_u) = 0$ implies $\dec(r_v)=0$ and that $\dec(r_u) = 1$ implies
$\dec(r_w)=1$. We are therefore left with two cases to consider.

Assume first that $r_u \geq 1- \delta = 1-\eps$. Then
\begin{equation}
  r_A \leq \frac{\eps}{1-\phi}+\eps =
  \frac{13}{3}\eps \leq \gamma \enspace , \label{eq:PURIFY-hard-ineq1}
\end{equation}
by the choice of $\gamma$, $\eps$, and $\phi$. It follows that
$r_v \geq 1-\eps=1-\delta$, thereby proving the remaining part of
(\ref{eq:PURIFYreq1}).  Assume next that $r_u \leq \gamma$. Then
$r_B \geq 1-\eps$ and thus
\begin{equation}
  r_w \leq \frac{\eps}{1-\eta}+\eps \leq \gamma \enspace , \label{eq:PURIFY-hard-ineq2}
\end{equation}
again by the choice of $\gamma$, $\eps$, and $\phi$, thereby proving
the remaining part of (\ref{eq:PURIFY-hard-ineq2}) as well, and
thereby completing the proof.

For clarity we evaluate the quantities in (\ref{eq:PURIFY-hard-ineq1})
and (\ref{eq:PURIFY-hard-ineq2}) numerically, and find that
$\frac{13}{3}\eps \approx 0.438$ and that
$\frac{\eps}{1-\eta}+\eps \approx 0.537$, while recalling that
$\gamma \approx 0.551$.

\paragraph{Completing the Reduction} 
Given the definition and analysis of the gadgets for \NAND\ and
\PURIFY\ gadget above, a reduction from $\PureCircuit$ to
\epsCDSClearing\ having $\eps=5-2\sqrt{6}$ is now immediate. Each gate
of the given $\PureCircuit$ instance is replaced by the corresponding
financial network gadget, and joining these together gives the
resulting financial network of the reduction. An $\eps$-approximate
clearing recovery rate vector $r$ for this financial network now
results in a solution to the $\PureCircuit$ instance given directly by
$\dec(r)$.

\section{$\ExistR$ Completeness Results}
In this section we consider two natural decision problems for
financial networks with CDSs, and show them to be $\ExistR$
complete. In the presence of default costs, a clearing recovery rate
vector is not guaranteed to exist. Define
$\CDSHasClearing{(\alpha,\beta)}$ to be the problem of deciding if a
given non-degenerate financial network with default costs $\alpha$ and
$\beta$ has a clearing recovery rate vector. The problem is clearly
contained in $\ExistR$ by the characterization of $\ExistR$ by BSS
machines~\cite{BurgisserCucker2009-Exotic-Quantifiers}, and we prove a
matching hardness result.
\begin{theorem}
  The problem $\CDSHasClearing{(\alpha,1)}$ is $\ExistR$-complete for
  all $\alpha<1$.
  \label{thm:CDSHasClearing}
\end{theorem}
Schuldenzucker~et~al.~\cite[Theorem~1]{SchuldenzuckerSB2019-Complexity}
proved \NP-hardness when either $\alpha<1$ or $\beta<1$ that holds a
\emph{gap} version of the problem, which is the appropriate hardness
notion in relation to $\eps$-approximate clearing. Our result
addresses the setting of exact clearing, and here strengthens the
result of \NP-hardness to \ExistR-hardness in the case of $\alpha<1$,
thereby completely settling the complexity of the problem. We leave
the problem of determining the precise computational complexity for
the case of $\beta<1$ as an interesting open problem. We note that
when $\beta<1$ the input gadget and the arithmetic gadgets break.

In the case of no default costs, although a clearing recovery rate
vector is guaranteed to exist, it will typically not be unique. In
such cases a regulator might have interest in the default or solvency
of a particular bank or set of banks. Define \CDSCanDefault\ to be the
problem of deciding if there exist a clearing recovery rate vector for
which a given bank defaults, and similarly define \CDSCanSurvive\ to
be the problem of deciding if there exist a clearing recovery rate
vector for which a given bank remains solvent. Both problems are again
clearly in \ExistR, and we complement this with a matching hardness
result for the latter problem.
\begin{theorem}
  The problem $\CDSCanSurvive$ is $\ExistR$-complete.
  \label{thm:CDSCanSurvive}
\end{theorem}
Schuldenzucker~et~al.~\cite[Theorem~1]{SchuldenzuckerSB2019-Complexity}
proved that gap versions of both problems are \NP-hard. For exact
clearing we settle the complexity of the problem \CDSCanSurvive\ and
leaving the the precise computational complexity of the
problem~\CDSCanDefault\ as another interesting open problem.

\subsection{Evaluating Polynomials in Financial Networks}
\label{sec:compute-polys}
In order to prove \ExistR-hardness by reduction from $\FEAS(\ncube)$
we shall build financial networks for evaluating a given
polynomial. Let $p$ be a degree~$4$ polynomial in $n$ variables
with~$s \geq 1$ monomials. By scaling we may, without loss of
generality, assume that all coefficients are bounded in magnitude by
$1/s$. Let us split $p$ into two polynomials $p^+$ and $p^-$, letting
$p^+$ be the sum of the monomials of $p$ with positive coefficients
and $p^-$ be the sum of the monomials of $p$ with negative
coefficients, but having the sign of its coefficients flipped. Thus
$p^+$ and $p^-$ are polynomials with no negative coefficients such
that $p(x)=p^+(x)-p^-(x)$. Also, by assumption
$p^+(x),p^-(x) \in [0,1]$ for all $x \in \ncube$.

Inputs ranging over $[0,1]$ are modeled by the \emph{input gadget}
illustrated in Figure~\ref{fig:Input}, having no inputs and having
single bank~$u$. A clearing recovery rate vector is $r_u=r_x$ for any
$r_x \in [0,1]$, thereby modeling the \emph{free variable}
$x \in [0,1]$.  This simple financial network was given already by
Eisenberg and Noe~\cite[Appendix~2]{EisenbergNoe2001-Systemic-Risk} as
an example with multiple clearing recovery rate vectors. One could
argue that for this particular financial network, the only reasonable
clearing recovery rate vector would be $r_u=r_x=1$. But on the other
hand it seems quite difficult to exclude more complicated financial
networks exhibiting a similar behavior in a natural way, or even just
in an algorithmically efficient way, and we therefore find it
acceptable as a gadget.
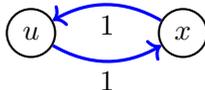
\begin{figure}[h!]
  \centering
  \begin{tikzpicture}
  
  \begin{scope}[every node/.style={circle,thick,draw}]
    \node (u) at (0,0) {$u$};
    \node (x) at (2, 0) {$x$};
  \end{scope}
  
  \begin{scope}[ every edge/.style={draw=blue,very thick}]
    \path [->] (x) edge[bend right] node[below] {$1$} (u);
    \path [->] (u) edge[bend right] node[below] {$1$} (x);
  \end{scope}
\end{tikzpicture} 

  \caption{Input gadget.}
  \label{fig:Input}
\end{figure}
The previous works of
Schuldenzucker~et~al.~\cite{SchuldenzuckerSB2019-Complexity} and
Ioannidis~et~al.~\cite{IoannidisKV2022-Strong-Approximations}
constructed financial network gadgets for arithmetic operations. In
particular we require a \emph{constant gadget} with the fixed output
$\zeta$, for any $\zeta \in [0,1]$, a \emph{sum gadget} computing
$\trunczo{u+v}$, a \emph{difference gadget} computing $\trunczo{u-v}$,
and a \emph{multiplication gadget} computing $u\cdot v$. We
furthermore require that these operate correctly in the setting of
$\alpha<1$. A multiplication gadget is given by
Ioannidis~et~al.~\cite{IoannidisKV2022-Strong-Approximations} and all
other gadgets are given by
Schuldenzucker~et~al.~\cite{SchuldenzuckerSB2019-Complexity}. For
completeness we describe the gadgets we use in
Appendix~\ref{sec:arithmetic-gadgets}. Combining all these we have the
following.
\begin{proposition}
  There is a polynomial time algorithm that given a degree~$4$
  polynomial $p$ in $n$ variables with $s\geq 1$ monomials and having
  coefficients of magnitude at most $1/s$, computes a financial
  network with an output gate $o$ such there is a one to one
  correspondence between clearing recovery rate vectors and the set
  $\ncube$. Furthermore the clearing recovery rate vector
  corresponding to $x \in \ncube$ has $r_o = \abs{p(x)}$.
    \label{prop:polygadget}  
\end{proposition}
\begin{proof}
  We use $n$~copies of the input gadget, one for each variable~$x_i$,
  $i=1,\dots,n$. Next we compute the polynomials $p^+$ and $p^-$ on
  input $x$ as follows.  Each coefficient of these is computed by a
  constant gadget. All monomials can then be computed using
  multiplication gadgets, and using addition gadgets we compute
  $p^+(x)$ and $p^-(x)$. Using difference gadgets we compute
  $\trunczo{p^+(x)-p^-(x)}$ and $\trunczo{p^-(x)-p^+(x)}$ and by a
  final addition gate we compute
  $\abs{p(x)} = \trunczo{p^+(x)-p^-(x)} + \trunczo{p^-(x)-p^+(x)}$.
\end{proof}

\subsection{Proof of Theorem~\ref{thm:CDSHasClearing}}

For our reduction we need a special case of an \emph{infeasibility
  gadget} by
Schuldenzucker~et~al.~\cite{SchuldenzuckerSB2019-Complexity}. We state its
properties here and for completeness describe the gadget in
Appendix~\ref{sec:infeasibility}.
\begin{lemma}[Infeasibility gadget
  {\cite[Lemma~5]{SchuldenzuckerSB2019-Complexity}}]
  \label{lem:infeas}
  There exists a financial network gadget with input bank $u$, such that
  \begin{itemize}
  \item If $r_u \geq \frac{3}{4}$, a clearing recovery rate vector exists.
  \item If $r_u \leq \frac{1}{4}$, no clearing recovery rate vector
    exists.
  \end{itemize}
\end{lemma}
To show $\ExistR$ hardness of \CDSHasClearing\ we give a
reduction from $\FEAS(\ncube)$. Following
Proposition~\ref{prop:polygadget} we build a financial network with
input gadgets selecting $x \in \ncube$ followed by computation of
$\abs{p(x)}$. We wish to distinguish between the two cases,
$\abs{p(x)}=0$ and $\abs{p(x)}>0$, by existence of a clearing recovery
rate vector. We create a ``gap'' between these cases using a
discontinuity gadget.
\begin{lemma}[Discontinuity gadget]
\label{lem:discont}
There exists a financial network gadget with input bank $u$ and output bank $v$, such that $r_v=1$ if $r_u=0$ and $r_v=(\alpha+1-r_u)/2$ if $r_u>0$.
\end{lemma}
\begin{proof}
  We wish to argue that the financial network gadget in Figure
  \ref{fig:discont} has this property.  First, assume $r_u = 0$. Then
  $p_{s,v} = c_{s, v}^u(1 - r_u) = 1$ and $v$ can pay all its
  liabilities, thus $r_v = 1$.  Now assume $r_u > 0$. Then $v$ is in
  default, and thus $r_v = (\alpha e_v + (1 - r_u))/l_v = (\alpha + 1 - r_u)/2$.
\end{proof}
\begin{figure}[h!]
    \centering
    \begin{tikzpicture}
  
  \begin{scope}[every node/.style={circle,thick,draw}]
    \node (u) at (1,2) [dashed] {$u$};
    \node (s) at (0,1) {$s$};
    \node (v) at (2,1) {$v$};
    \node (t) at (4,1) {$t$};
  \end{scope}
  
  \begin{scope}[ every edge/.style={draw=blue,very thick}]
    \path [->] (v) edge node[below] {$2$} (t);
  \end{scope}
  
  \begin{scope}[every edge/.style={draw=orange,very thick}]
    \path [->] (s) edge node[below] {$1$} (v);
  \end{scope}
  
  \begin{scope}[every edge/.style={draw=orange,dashed,very thick}]
    \path[-] (u) edge ($(s) !.5! (v)$);
  \end{scope}

  \begin{scope}[xshift=0.4cm, yshift=-0.4cm, every node/.style={rectangle, draw=black, fill=white}]   
    \node (ew) at (2,1) {$1$};
  \end{scope}
\end{tikzpicture}    

    \caption{Discontinuity gadget.}
    \label{fig:discont}
\end{figure}
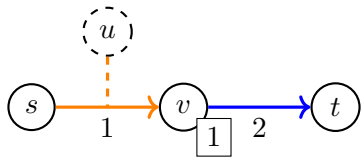
Using a discontinuity gadget with input $\abs{p(x)}$ and output bank
$v$ means that we must now distinguish between the cases $r_v=1$ and
$r_v (\alpha + 1 - \abs{p(x)})/2 \leq (1+\alpha)/2$. In order
distinguish between these using the infeasibility gadget of
Lemma~\ref{lem:infeas} we perform a sequence of~$k$ repeated squaring
operations using the multiplication gadget. Since
\[
  \left(\frac{1+\alpha}{2}\right)^{2^k} = \left(1- \frac{1-\alpha}{2}\right)^{2^k} < \exp\left(-(1-\alpha)2^{k-1}\right)
\]
we can pick $k=O(\log(1/(1-\alpha)))$ such that $(r_v)^{2^k} \leq 1/4$
whenever $r_v \leq (1+\alpha)/2$. For the final step of the
construction we use $(r_v)^{2^k}$ as input to the infeasibility
gadget. An abstract overview of the complete financial network created
by reduction is shown in Figure~\ref{fig:feasToCDS}. By construction
we have that there exists $x \in \ncube$ such that $p(x)=0$ if and only
if this financial network has a clearing recovery rate vector, thereby
completing the proof.

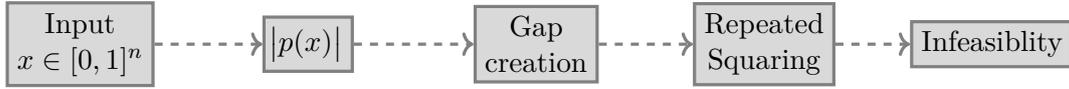
\begin{figure}[h!]
    \centering
    \begin{tikzpicture}
  
  \begin{scope}[every node/.style={draw=gray, fill=gray!30, align=center, rectangle, very thick}]
    \node (x) at (0,0) {Input \\ $x \in [0,1]^n$};
    \node (p) at (3,0) {$\abs{p(x)}$};
    \node (gap) at (6,0) {Gap\\creation};
    \node (sq) at (9,0) {Repeated\\Squaring};
    \node (infeas) at (12,0) {Infeasiblity};  
    
  \end{scope}

  \begin{scope}[every edge/.style={draw=gray,dashed,very thick}]
    \path[->] (x) edge (p);
    \path[->] (p) edge (gap);
    \path[->] (gap) edge (sq);
    \path[->] (sq) edge (infeas);
    
  \end{scope}

\end{tikzpicture}    

    \caption{An overview of the reduction to
      $\CDSHasClearing(\alpha,1)$.}
    \label{fig:feasToCDS}
\end{figure}

\subsection{Proof of Theorem~\ref{thm:CDSCanSurvive}}

The proof of Theorem~\ref{thm:CDSCanSurvive} starts in the same way as
the proof of Theorem~\ref{thm:CDSHasClearing}. We construct a
reduction from $\FEAS(\ncube)$ and from a given polynomial $p$,
following Proposition~\ref{prop:polygadget}, we build an financial
network with input gadgets for selecting $x \in \ncube$ followed by
computation of $\abs{p(x)}$. For the remainder of the construction we
use a simple financial network as shown in
Figure~\ref{fig:defaultset}.  We now wish to argue that there exists a
clearing recovery rate vector such that bank~$b$ is solvent if and
only if there exists $x \in \ncube$ such that $p(x)=0$.

We have that $r_b=1$ if and only bank~$b$ receives a payment of~$1$,
which happens if and only if $r_A=0$, which in turn holds if and only
if $\abs{p(x)}=0$, thereby completing the proof.
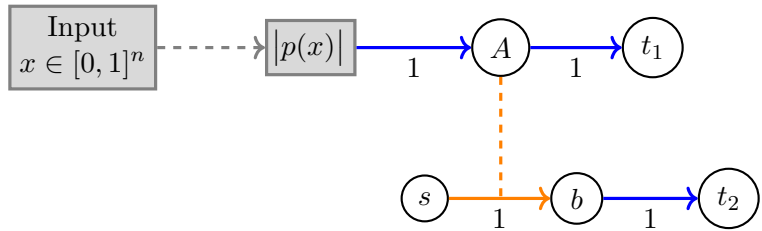
\begin{figure}[h!]
    \centering
    \begin{tikzpicture}
  
  \begin{scope}[every node/.style={draw=gray, fill=gray!30, align=center, rectangle, very thick}]
    \node (x) at (2.5,1) {Input \\ $x \in [0,1]^n$};
    \node (p) at (5.5,1) {$\abs{p(x)}$};
  \end{scope}

  \begin{scope}[every node/.style={circle, thick, draw, fill=white}]
    \node (A) at (8,1) {$A$};
    \node (t1) at (10,1) {$t_1$};
    \node (b) at (7+2,-1) {$b$};
    \node (t2)at (9+2,-1) {$t_2$};
    \node (s) at (5+2,-1) {$s$};
  \end{scope}

  \begin{scope}[every edge/.style={draw=gray,dashed,very thick}]
    \path[->] (x) edge (p);
  \end{scope}
  
  \begin{scope}[ every edge/.style={draw=blue,very thick}]
    \path[->] (p) edge node[below]{$1$}(A);
    \path[->] (A) edge node[below]{1} (t1);
    \path[->] (b) edge node[below]{1} (t2);
  \end{scope}

  \begin{scope}[every edge/.style={draw=orange,very thick}]
    \path [->] (s) edge node[below] {1} (b);
  \end{scope}
  
  \begin{scope}[every edge/.style={draw=orange,dashed,very thick}]
    \path[-] (A) edge  ($(s) !.5! (b)$);
  \end{scope}

\end{tikzpicture}    

    \caption{An overview the reduction to \CDSCanSurvive.}
    \label{fig:defaultset}
\end{figure}


\bibliographystyle{abbrvurl}
\bibliography{financialcomplexity}

\appendix

\section{Non-existence of a clearing recovery rate vector in
  a degenerate financial networks}
\label{sec:degeneracy}

In Figure~\ref{fig:degenerate} we give a relatively simple example of
a financial network, violating the non-degeneracy assumption, without
a clearing recovery vector. By construction of the network and
Definition~\ref{def:Clearing}, if $r$ is a clearing recovery rate
vector the following three equations must hold.
\begin{align}
  r_A & = \min(1,2(1-r_B))\enspace , &
  r_B & = \begin{cases}1 & \text{ if } r_C=1\\\frac{1}{2} & \text{ if } r_C<1\end{cases}\enspace , &
  r_C & = \begin{cases}1 & \text{ if } r_A=1\\0 & \text{ if } r_A<1\end{cases}\enspace .
\end{align}

With these we can rule out a clearing recovery vector. Suppose first
that $r_C=1$. The second and first equation then give that $r_B=1$ and
$r_A=0$. But then the third equation gives the contradiction that
$r_C=0$. Suppose next that $r_C=0$. The second and first equation then
give that $r_B=\frac{1}{2}$ and $r_A=1$. But then the third equation
gives the contradiction that $r_C=1$.

\begin{figure}[h!]
  \centering
  \begin{tikzpicture}
  \begin{scope}[every node/.style={circle, thick, draw}]   
    \node (f) at  (0,0) {f};
    \node (B) at  (2,0) {B};
    \node (g) at  (4,0) {g};
    \node (d) at  (2,1.5) {d};
    \node (A) at  (6,1.5) {A};
    \node (e) at  (8,1.5) {e};
    \node (C) at  (2,-1.5) {C};
    \node (h) at  (6,-1.5) {h};
  \end{scope}
  
  \begin{scope}[ every edge/.style={draw=blue,very thick}]
    \path [->] (A) edge node[above] {$1$} (e);
  \end{scope}
  
  \begin{scope}[every edge/.style={draw=orange,very thick}]
    \path [->] (f) edge node[above] {$1$} (B);
    \path [->] (B) edge node[above] {$2$} (g);
    \path [->] (d) edge node[above] {$2$} (A);
    \path [->] (C) edge node[below] {$1$} (h);
  \end{scope}
  
  \begin{scope}[every edge/.style={draw=orange,dashed,very thick}]
    \path[-] (B) edge ($(d) !.5! (A)$);
    \path[-] (C) edge ($(f) !.5! (B)$);
    \path[-] (C) edge ($(B) !.5! (g)$);
    \path[-] (A) edge ($(C) !.5! (h)$);
  \end{scope}

  \begin{scope}[xshift=0.4cm, yshift=-0.4cm, every node/.style={rectangle, draw=black, fill=white}]
    \node (ed) at (2,1.5) {$2$};
    \node (ef) at (0,0) {$1$};
  \end{scope}
  
\end{tikzpicture}    

    \caption{Degenerate Financial Network.}
    \label{fig:degenerate}
\end{figure}
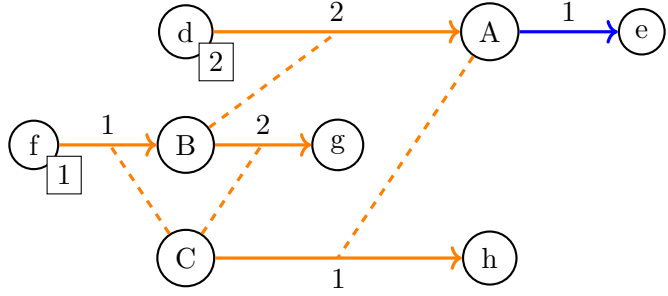

As an alternative to making an assumption of non-degeneracy, one could
alternatively leave the recovery rate of a bank without liabilities be
unconstrained. Doing so would ensure that a clearing recovery rate
vector would always exists. In the example above one such clearing recovery
rate vector would be given by
$(r_A,r_B,r_C)=\left(1,\frac{1}{2},1\right)$, where bank~$B$ is given
recovery rate $\frac{1}{2}$ despite having no liabilities. Besides
ensuring existence with a less restrictive definition, it is easy to
show that the task of computing such a clearing recovery rate vector
is in $\FIXP$ using the recent framework of Filos-Ratsikas et
al.~\cite{Filos-RatsikasH20XX-FIXP-OPT}. On the other hand it is not
known whether computing an \eps-approximate clearing recovery rate
vector is in $\PPAD$ without the non-degeneracy
assumption~\cite{SchuldenzuckerSB2019-Complexity}.

\section{\PPAD-membership of \epsCDSClearing}
\label{sec:PPAD-membership}

Here we prove \PPAD-membership for the version of the \epsCDSClearing\
problem given by the notion of approximation defined in
Definition~\ref{def:eps-approx}.

\begin{proposition}
  The problem \epsCDSClearing\ is in \PPAD.
\end{proposition}
\begin{proof}
  Define the auxiliary functions $G \colon [0,1+\eps]^N \to [0,1+\eps]^N$ and $g \colon [0,1+\eps]^N \to [0,1+\eps]^N$  by
  \begin{equation}
    G(r)_i =
    \begin{cases}
      1+\eps & \text{ if } a_i(\trunczo{r}) \geq (1+\eps)l_i(\trunczo{r}) \\
      \frac{a_i(\trunczo{r})}{l_i(\trunczo{r})} & \text{ if } a_i(\trunczo{r}) < (1+\eps)l_i(\trunczo{r})
    \end{cases} 
    \label{eq:G-function}
  \end{equation}
and
\begin{equation}  
  g(r)_i = \frac{a_i(\trunczo{r})}{\max\left(\frac{a_i(\trunczo{r})}{1+\eps},l_i(\trunczo{r})\right)} \enspace .
  \label{eq:g-function}
\end{equation}
The function $g$ is well-defined on the entire domain by the
assumption of non-degeneracy.  Note also that by definition of $G$ and
$g$ we have that $G(r)=G(\trunczo{r})$ and $g(r)=g(\trunczo{r})$ for
all $r \in [0,1+\eps]^N$.

By the assumption of non-degeneracy we also have that the function $g$ is
polynomially continuous and polynomially computable, and computing an
\eps-almost fixed point of the function $g$ (i.e.\ a point $r$ such
that $g(r)=r\pm \eps$) is in \PPAD\
by~\cite[Proposition~2.2]{EtessamiYannakakis2010-FIXP}. Computing a
solution to \epsCDSClearing\ now simply consists of computing an
\eps-almost fixed point $r$ of $g$ and letting $r'=\trunczo{r}$ be the
resulting recovery rate vector. Next we shall prove that $r'$ is an
\eps-almost fixed point of $f$ and furthermore satisfies that
$a_i(r') \geq (1+\eps)l_i(r') \implies r'_i=1$ for all~$i \in N$.

Observe first that if $l_i(r')=0$, then $G(r')_i=g(r')_i=1+\eps$ and
thus $r'_i = 1$. Since also $F(r')_i=f(r')_i=1$ we have that both
conditions are trivially satisfied.

Assume now that $l_i(r')\neq 0$. We then have that
$G(r')_i=\min\left(1+\eps,\frac{a_i(r')}{l_i(r')}\right)=g(r')_i$. First
note that if $a_i(r')\geq (1+\eps)l_i(r')$ then $g(r')_i=1+\eps$ which
implies that $r'_i=1$.  Since also $F(r')_i=f(r')_i=1$ both conditions
are again satisfied. Consider now the remaining case of
$a_i(r') < (1+\eps)l_i(r')$. Then we have
$g_i(r')=\frac{a_i(r')}{l_i(r')}$ and thus
$r_i = \frac{a_i(r')}{l_i(r')} \pm \eps$, from which it follows that
$r'_i = \trunczo{\frac{a_i(r')}{l_i(r')}} \pm \eps$. Since
$f(r')=\min\left(1,\frac{a_i(r')}{l_i(r)}\right)$ we again have that the
conditions are satisfied.

\end{proof}

\section{Financial Network Gadgets for Arithmetic}
\label{sec:arithmetic-gadgets}
In this section we describe and analyze financial network gadgets for
arithmetic operations. We note that no banks (except implicitly, source
and sink banks) have external assets, meaning that they operate
correctly in the setting of $\alpha<1$.

\begin{lemma}[Constant~gadget {\cite[Lemma~6]{SchuldenzuckerSB2019-Complexity}}]
\label{lem:const}
For any constant $\zeta \in [0,1]$ there exists a financial network
gadget with no input bank and output bank $v$, such that
$r_v = \zeta$.
\end{lemma}
\begin{proof}
  We argue that the financial network gadget in Figure~\ref{fig:const}
  satisfies the statement. We have $r_v = \trunczo{a_v(r)} = a_v(r) = \zeta$.
\end{proof}
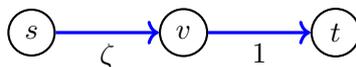
\begin{figure}[h!]
  \centering
  \begin{tikzpicture}
  
  \begin{scope}[every node/.style={circle,thick,draw}]
    \node (s) at (0,1) {$s$};
    \node (v) at (2,1) {$v$};
    \node (t) at (4,1) {$t$};
  \end{scope}
  
  \begin{scope}[ every edge/.style={draw=blue,very thick}]
    \path [->] (v) edge node[below] {$1$} (t);
    \path [->] (s) edge node[below] {$\zeta$} (v); 
  \end{scope} 
\end{tikzpicture}

  \caption{Constant gadget.}
  \label{fig:const}
\end{figure}

\begin{lemma}[Sum gadget {\cite[Lemma~8]{SchuldenzuckerSB2019-Complexity}}]
\label{lem:plus}
There exists a financial network gadget with input banks $u$ and $v$ and
output bank $w$, such that $r_w = \trunczo{r_u + r_v}$.
\end{lemma}
\begin{proof}
  We argue that the financial network gadget in Figure~\ref{fig:plus}
  satisfies the statement. We have
    \begin{align*}
        r_w &= \trunczo{a_w(r)}  \\
            &= \trunczo{1-r_A + 1 - r_B} \\
            &= \trunczo{1-(1 - r_u) + 1 - (1 - r_v)} \\
            &= \trunczo{r_u  + r_v}.
    \end{align*}
\end{proof}
\begin{figure}[h!]
  \centering \begin{tikzpicture}
  \begin{scope}[every node/.style={circle, thick, draw}]
    \node (u) at (1,3) [dashed] {$u$};
    \node (v) at (7,3) [dashed] {$v$};
    \node (s1) at (0,2) {$s_1$};
    \node (A) at (2,2) {$A$};
    \node (t1) at (4,2) {$t_1$};

    \node (s2) at (6+0,2) {$s_2$};
    \node (B) at  (6+2,2) {$B$};
    \node (t2) at (6+4,2) {$t_2$};

    \node (s3) at (2+1,0) {$s_3$};
    \node (w) at  (2+3,0) {$w$};
    \node (t3) at (2+3,-2) {$t_3$};
    \node (s4) at (2+5,0) {$s_4$};
  \end{scope}
  
  \begin{scope}[ every edge/.style={draw=blue,very thick}]
    \path [->] (A) edge node[below] {$1$} (t1);
    \path [->] (B) edge node[below] {$1$} (t2);
    \path [->] (w) edge node[right] {$1$} (t3);
  \end{scope}
  
  \begin{scope}[every edge/.style={draw=orange,very thick}]
    \path [->] (s1) edge node[below] {$1$} (A);
    \path [->] (s2) edge node[below] {$1$} (B);
    \path [->] (s3) edge node[below] {$1$} (w);
    \path [->] (s4) edge node[below] {$1$} (w);
  \end{scope}
  
  \begin{scope}[every edge/.style={draw=orange,dashed,very thick}]
    \path[-] (u) edge ($(s1) !.5! (A)$);
    \path[-] (A) edge ($(s3) !.5! (w)$);
    \path[-] (v) edge ($(s2) !.5! (B)$);
    \path[-] (B) edge ($(s4) !.5! (w)$);
  \end{scope}
\end{tikzpicture}    
  \caption{Sum gadget.}
  \label{fig:plus}
\end{figure}
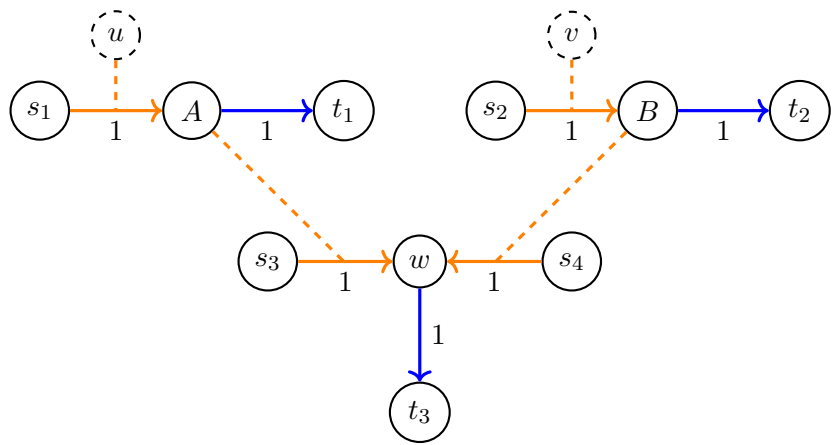

\begin{lemma}[Difference gadget {\cite[Lemma~9]{SchuldenzuckerSB2019-Complexity}}]
    \label{lem:minus}
    There exists a financial network gadget with input bank $u, v$ and output bank $w$ such that $r_w = \trunczo{r_v - r_u}$.
\end{lemma}
\begin{proof}
    We argue that the financial network gadget in Figure~\ref{fig:difference} satisfies the statement. We have
    \begin{align*}
        r_w &= 1 - r_B \\
        &= 1 - \trunczo{(1 - r_A) + (1 - r_v)}\\
        &= 1 - \trunczo{1 - (1-r_u) + (1 - r_v)}\\
        &= 1 - \trunczo{1 - (r_v - r_u)}\\
        & = \trunczo{r_v - r_u} 
    \end{align*}
\end{proof}
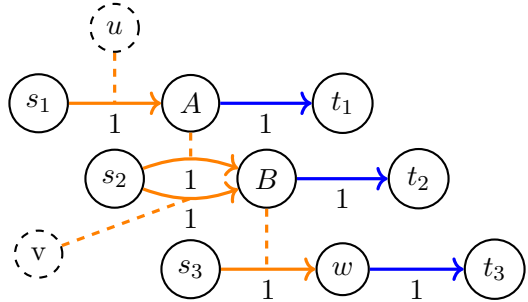
\begin{figure}[h!]
    \centering
    \begin{tikzpicture}
  
  \begin{scope}[every node/.style={circle, thick, draw}]
    \node (u) at (1,2) [dashed] {$u$};
    
    \node (s1) at (0,1) {$s_1$};
    \node (A) at  (2,1) {$A$};
    \node (t1) at (4,1) {$t_1$};
    \node (s2) at (1,0) {$s_2$};
    \node (C) at  (3,0) {$B$};
    \node (t2) at (5,0) {$t_2$};
    \node (s3) at (2,-1.2) {$s_3$};
    \node (w) at  (4,-1.2) {$w$};
    \node (t3) at (6,-1.2) {$t_3$};
    
    \node (v) at  (0, -1) [dashed] {v};

  \end{scope}
  
  \begin{scope}[ every edge/.style={draw=blue,very thick}]
    \path [->] (A) edge node[below] {$1$} (t1);

    \path [->] (C) edge node[below] {$1$} (t2);

    \path [->] (w) edge node[below] {$1$} (t3);

  \end{scope}
  
  \begin{scope}[every edge/.style={draw=orange,very thick}]
    \path [->] (s1) edge node[below] {$1$} (A);
    \path [->] (s2) edge [bend left=20]  coordinate[pos=.5] (ac)  node[below] {$1$} (C);
    \path [->] (s2) edge [bend right=20] coordinate[pos=.5] (vc)  node[below] {$1$} (C);
    \path [->] (s3) edge node[below] {$1$} (w);
  \end{scope}
  
  \begin{scope}[every edge/.style={draw=orange,dashed,very thick}]
    \path[-] (u) edge ($(s1) !.5! (A)$);
    \path[-] (A) edge (ac);
    \path[-] (v) edge  (vc);
    \path[-] (C) edge ($(s3) !.5! (w)$);
  \end{scope}
\end{tikzpicture}    

    \caption{Difference gadget.}
    \label{fig:difference}
\end{figure}
For multiplication,
Ioannides~\cite{IoannidisKV2022-Strong-Approximations} give a simple,
but \emph{degenerate}, financial network gadget, and afterwards modify
it to a non-degenerate gadget computing the expression
$(r_u(1+r_v))/4$. We give a slightly different and arguably simpler
gadget below, that also avoids having banks (other than source or sink
banks) with external assets.
\begin{lemma}[Multiplication gadget]
    \label{lem:mult}
    There exists a financial network gadget with input banks $u$ and
    $v$ and output bank $w$, such that $r_w = r_ur_v/2$.
\end{lemma}
\begin{proof}
  We argue that the financial network gadget in Figure~\ref{fig:mult}
  satisfies the statement. First we have that $r_A=1-r_u$ and $r_B=1-r_v$. It follows that
  \[
    r_C=\frac{1-r_A}{1+(1-r_v)+(1-r_B)}=\frac{r_u}{1+(1-r_v)+r_v}=\frac{r_u}{2} \enspace ,
  \]
  and finally we have
  \[
    r_w = r_C(1-r_B)=r_ur_v/2 \enspace .
  \]
\end{proof}
\begin{figure}[h!]
    \centering
    \begin{tikzpicture}
  \begin{scope}[every node/.style={circle, thick, draw}]
    \node (u) at (0,3) [dashed] {$u$};
    \node (v) at (7,3) [dashed] {$v$};
    
    \node (s1) at (-1+0,2) {$s_1$};
    \node (A) at  (-1+2,2) {$A$};
    \node (t1) at (-1+4,2) {$t_1$};

    \node (s2) at (7+0,2) {$s_2$};
    \node (B) at  (7+2,2) {$B$};
    \node (t2) at (7+4,2) {$t_2$};

    \node (s3) at (2+1,0) {$s_3$};
    \node (C) at  (2+3,0) {$C$};
    \node (t3) at (2+7,0) {$t_3$};
    \node (w) at  (2+5,0) {$w$};
    \node (t4) at (2+3, 2) {$t_4$};

  \end{scope}
  
  \begin{scope}[ every edge/.style={draw=blue,very thick}]
    \path [->] (A) edge node[below] {$1$} (t1);
    \path [->] (B) edge node[below] {$1$} (t2);
    \path [->] (w) edge node[below] {$1$} (t3);

    \path [->] (C) edge[bend left] node[left] {$1$} (t4);
  \end{scope}
  
  \begin{scope}[every edge/.style={draw=orange,very thick}]
    \path [->] (s1) edge node[below] {$1$} (A);
    \path [->] (s2) edge node[below] {$1$} (B);
    \path [->] (s3) edge node[below] {$1$} (C);
    \path [->] (C) edge node[below] {$1$} (w);
    \path [->] (C) edge[bend right] coordinate[pos=.5] (midCt4) node[right] {$1$} (t4);
  \end{scope}
  
  \begin{scope}[every edge/.style={draw=orange,dashed,very thick}]
    \path[-] (u) edge ($(s1) !.5! (A)$);
    \path[-] (A) edge ($(s3) !.5! (C)$);
    \path[-] (v) edge ($(s2) !.5! (B)$);
    \path[-] (B) edge ($(C) !.5! (w)$);
    \path[-] (v) edge (midCt4);
  \end{scope}
\end{tikzpicture}    

    \caption{Multiplication gadget.}
    \label{fig:mult}
\end{figure}
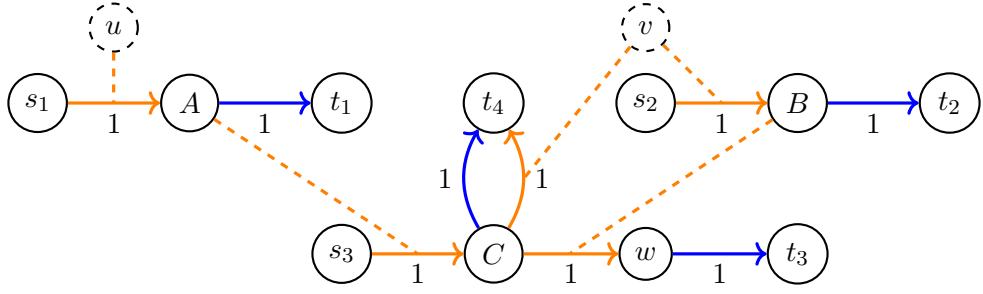
We may use the addition gadget of Lemma~\ref{lem:plus} together with
the multiplication gadget of Lemma~\ref{lem:mult} to get a ``proper''
multiplication gadget computing the expression $r_ur_v$.

\section{Infeasibility gadget}
\label{sec:infeasibility}
In this section we describe a special case of the infeasibility gadget
constructed by by
Schuldenzucker~et~al.~\cite{SchuldenzuckerSB2019-Complexity}. The
gadget is based on a \emph{cut-off gadget} and an \emph{OR gadget}.

\begin{lemma}[Cut-off Gadget {\cite[Lemma~3]{SchuldenzuckerSB2019-Complexity}}]
  \label{lem:cutoff}
  There exists a financial network gadget with input bank $u$ and
  output bank $v$, such that
\begin{align*}
    r_u \leq K &\implies r_v = 0\\
    r_u \geq L &\implies r_v = 1.
\end{align*}
\end{lemma}
\begin{proof}
  We argue that the financial network in Figure \ref{fig:cutoff}
  satisfies the statement.  Assume $r_u \leq K$. Then
  $p_{s,A} \geq \frac{1}{1-K}(1-K) = 1$, so $A$ does not default,
  causing $r_v = 0$.  Now assume $r_u \geq L$. This means
  $r_A \leq \frac{1 - L}{1 - K}$, so
  $r_v \geq \frac{1 - K}{L - K}  \left(1 - \frac{1 - L}{1 -
      K}\right) = 1.$
\end{proof}
\begin{figure}[h!]
  \centering
  \begin{tikzpicture}
  
  \begin{scope}[every node/.style={circle,thick,draw}]
    \node (u) at (1,3) [dashed] {$u$};
    \node (s) at (0,1) {$s$};
    \node (A) at (2,2) {$A$};
    \node (t) at (4,1) {$t$};
    \node (v) at (2,0) {$v$};
  \end{scope}
  
  \begin{scope}[ every edge/.style={draw=blue,very thick}]
    \path [->] (v) edge node[below right] {$1$} (t);
    \path [->] (A) edge node[above right] {$1$} (t);
  \end{scope}
  
  \begin{scope}[every edge/.style={draw=orange,very thick}]
    \path [->] (s) edge node[above left] {$\frac{1}{1 - K}$} (A);
    \path [->] (s) edge node[below left] {$\frac{1 - K}{L - K}$} (v);
  \end{scope}
  
  \begin{scope}[every edge/.style={draw=orange,dashed,very thick}]
    \path[-] (u) edge ($(s) !.5! (A)$);
    \path[-] (A) edge ($(s) !.5! (v)$);
  \end{scope}
\end{tikzpicture}    

  \caption{Cut-off gadget.}
  \label{fig:cutoff}
\end{figure}
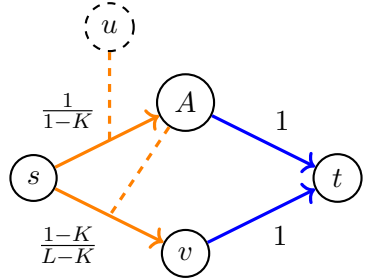

\begin{lemma}[OR Gadget \cite{SchuldenzuckerSB2019-Complexity}]
    There exists a financial network gadget with input banks $u, v$ and output bank $w$, such that \begin{align*}
        r_u \leq \frac{1}{4} \wedge r_v \leq \frac{1}{4} &\implies r_w = 0,\\
        r_u \geq \frac{3}{4} \vee r_v \geq \frac{3}{4} & \implies r_w = 1.
    \end{align*}
\end{lemma}
\begin{proof}
  First, we use two cut-off gadgets from Lemma \ref{lem:cutoff} with
  input banks $u, v$ and output banks $u^\prime, v^\prime$,
  respectively, and with $K = \frac{1}{4}$, $L = \frac{3}{4}$. Now,
  $r_{u^\prime} = 0$ if $r_u \leq \frac{1}{4}$ and $r_{u^\prime} = 1$
  if $r_u \geq \frac{3}{4}$. The same holds for $v, v^\prime$.  Next,
  we use $u^\prime, v^\prime$ as input banks to a sum gadget from
  Lemma \ref{lem:plus}. From this lemma, we get that
  \[
    r_w =
    \begin{cases}
      0 & \text{ if } r_{v^\prime} = r_{u^\prime} = 0\\
      1 & \text{ if } r_{v^\prime} = 1 \vee r_{u^\prime} = 1\\
    \end{cases}.
  \]
\end{proof}
With these gadgets in place we can now describe and analyze the
infeasibility gadget.
\begin{proof}[Proof of Lemma~\ref{lem:infeas}]
  We argue that the financial network in Figure~\ref{fig:infeas}
  satisfies the statement when the cut-off gadget has
  $K = \frac{3\alpha + 1}{4}$ and $L = \frac{\alpha + 3}{4}$.
    
  First let us assume $r_u \geq \frac{3}{4}$. Then the OR gadget with
  output bank $A$ has $r_A = 1$, so $r_B = \frac{4\alpha}{5}$, and
  $r_C = 0$ giving us a clearing recovery rate vector.
  Now, let us assume $r_u \leq \frac{1}{4}$. We wish to argue by
  contradiction that no clearing recovery rate vector exists.  If
  $r_B = 1$, then $r_A = 1$ per the cut-off and OR gadgets, and thus
  $p_{s, B} = 0$, so $B$ is in default. This is a contradiction.  If
  $r_B < 1$, then
    \begin{align*}
        r_B &= a'_B(r)\\
        &= \alpha e_B + p_{s, B} \\
        &= (1 - 1 + \alpha) e_B + p_{s, B}\\
        & = e_B + p_{s, B} - e_B(1 - \alpha) \\
        & = a_B(r) - e_B(1 - \alpha)  & (\text{since } a_B(r) = e_B + p_{s, B}) \\
        & < 1 - e_B(1 - \alpha)  & (\text{since } r_B < 1 \text{ we have } a_B < l_B = 1)\\
        & = 1 - \frac{4}{5}(1 - \alpha)\\
        & = \frac{4 \alpha + 1}{5}\\
        & \leq \frac{3 \alpha + 1}{4}\\
        & = K.
    \end{align*}
    Thus by the cut-off gadget $r_C = 0$, and thus $r_A = 0$ (since
    $r_u \leq \frac{1}{4}$). This means $p_{s, B} = \frac{4}{5}$, but
    this means $B$ is not in default. We again get a
    contradiction. Thus we conclude no clearing recovery rate vector
    exists.
\end{proof}
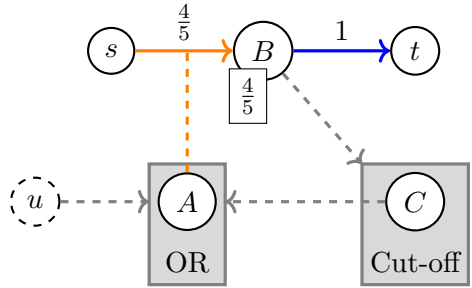
\begin{figure}[h!]
  \centering
  \begin{tikzpicture}
  
  \begin{scope}[every node/.style={circle, thick, draw, fill=white}]
    \node (u) at (-1,1.5) [dashed] {$u$};
    \node (s) at (0,3.5) {$s$};
    \node (B) at (2,3.5) {$B$};
    \node (t) at (4,3.5) {$t$};
    \node (C) at (4,1.5) {$C$};
    \node (A) at (1,1.5) {$A$};
  \end{scope}

  \begin{scope}[on background layer]
    \draw[on background layer,draw=gray, fill=gray!30, very thick] (0.5,0.4) rectangle (1.5,2);
    \draw[on background layer,draw=gray, fill=gray!30, very thick] (3.3,0.4) rectangle (4.7,2); 
  \end{scope}
  
  \begin{scope}[ every edge/.style={draw=blue,very thick}]
    \path [->] (B) edge node[above] {$1$} (t);
  \end{scope}
  
  \begin{scope}[every edge/.style={draw=orange,very thick}]
    \path [->] (s) edge node[above] {$\frac{4}{5}$} (B);
  \end{scope}
  
  \begin{scope}[every edge/.style={draw=orange,dashed,very thick}]
    \path[-] (A) edge  ($(s) !.5! (B)$);
  \end{scope}

  \begin{scope}[every edge/.style={draw=gray,dashed,very thick}]
    \path[->] (B) edge (3.3,2); 
    \path[->] (u) edge (0.5,1.5);
    \path[->] (C) edge (1.5,1.5);
  \end{scope}

  \begin{scope}[xshift=-0.2cm, yshift=-0.6cm, every node/.style={rectangle, draw=black, fill=white}]   
    \node (ew) at (2,3.5) {$\frac{4}{5}$};
  \end{scope}

  \begin{scope}
    \node (OR) at (1,0.7) {OR};
    \node (Cutoff) at (4, 0.7) {Cut-off};
  \end{scope}
\end{tikzpicture}    

  \caption{Infeasibility gadget. The cut-off gadget has input bank $B$ and output bank $C$. The OR gadget has input banks $u, C$ and output bank $A$.}
  \label{fig:infeas}
\end{figure}

\end{document}